\RequirePackage{cmap}
\documentclass[english,openary,12pt]{article}
\usepackage[cp1251]{inputenc}
\usepackage{cmap}
\usepackage[T2A]{fontenc}
\usepackage[english]{babel}
\usepackage{graphicx}
\DeclareGraphicsExtensions{.png}
\usepackage{floatflt}
\usepackage{amsfonts}                              
\usepackage{amsthm}                                
\usepackage{amsmath}
\usepackage{mathtools}                            
\usepackage{gensymb}
\usepackage{amssymb,amsfonts}
\usepackage[usenames]{color}
\usepackage{wrapfig}
\usepackage{url}
\usepackage{cmap}
\usepackage{graphicx}
\usepackage{multirow}
\usepackage{comment}
\usepackage{caption}
\usepackage{subcaption}
\usepackage{float}
\usepackage{diagbox}
\usepackage{enumerate}
\usepackage{titling}
\renewcommand{\emptyset}{\varnothing}

\makeatletter
\renewcommand\section{\@startsection {section}{1}{\z@}%
                                   {-3.5ex \@plus -1ex \@minus -.2ex}%
                                   {2.3ex \@plus.2ex}%
                                   {\normalfont\LARGE\bfseries}}
\renewcommand\subsection{\@startsection{subsection}{2}{\z@}%
                                     {-3.25ex\@plus -1ex \@minus -.2ex}%
                                     {1.5ex \@plus .2ex}%
                                     {\normalfont\Large\bfseries}}
\renewcommand\subsubsection{\@startsection{subsubsection}{3}{\z@}%
                                     {-3.25ex\@plus -1ex \@minus -.2ex}%
                                     {1.5ex \@plus .2ex}%
                                     {\normalfont\large\bfseries}}
\makeatother

\usepackage[left=1.65cm,right=1.65cm,
    top=1.57cm,bottom=1.57cm,bindingoffset=0cm]{geometry}

\usepackage{amsthm}
\usepackage{bm}

\theoremstyle{definition}
\newtheorem{remark}{Remark}
\newtheorem{conjecture}{Conjecture}
\newtheorem{example}{Example}
\newtheorem{definition}{Definition}
\newtheorem{Symbol}{Notation}
\newtheorem{theorem}{Theorem}

\newtheorem{corollary}{Corollary}
\newtheorem{proposition}{Proposition}

\newcommand{\ee}{\mathbb{E}}
\newcommand{\zz}{\varepsilon\mathbb{Z}^2}
\newcommand{\zp}{\varepsilon\mathbb{Z}_+}
\newcommand{\z}{\varepsilon\mathbb{Z}}
\newcommand{\Z}{\mathbb{Z}}

\DeclarePairedDelimiter\floor{\lfloor}{\rfloor}

\begin{document}

\author{Novikov Ivan}
\title{\vspace{-2.1cm}Feynman checkers: the probability to find an electron vanishes nowhere inside the light cone \vspace{-0.42cm}}
\date{}

\maketitle

\vspace{-1.5cm}

\begin{abstract}
We study Feynman checkers, the most elementary model of electron motion introduced by R.~Feynman. For the model, we prove that the probability to find an electron vanishes nowhere inside the light cone. We also prove several results on the average electron velocity. In addition, we present a lot of identities related to the model.

\end{abstract}

\section{Introduction}

This paper is on the most elementary model of one-dimensional electron motion that is known as ``Feynman checkers''. This model was introduced by R.~Feynman around the 1950s and published in 1965; see \cite[Problem 2.6]{Feynman}. Afterwards, a large amount of physical articles on the model appeared (see, for example, \cite{Quant,Jacobson,Konno,Narlikar,Ord}). But the first mathematical work \cite{main} on the subject appeared, apparently, only in 2020. We use the inessential modification of the model from \cite{Feynman} that was presented in \cite{main}. Many properties of our model are parallel to usual quantum mechanics: there are analogues of Dirac equation (Proposition~\ref{Dirac}), probability/charge conservation (Proposition~\ref{prob}), Klein--Gordon equation \cite[Proposition~7]{main}, Fourier integral \cite[Propositions~12-13]{main}, concentration of measure on the light cone \cite[Corollary~6]{main} etc. Other striking properties have sharp contrast with both continuum quantum theory and the classical random walks: the (essentially) maximal electron velocity is strictly less than the speed of light \cite[Theorem~1]{Quant}, \cite[Theorem~1(B)]{main}; adding absorbing boundary increases the probability of returning to the initial point \cite[Theorems~8 and 10]{Quant}.

It should be noted that Feynman checkers almost completely identical to \emph{one-dimensional quantum walk} and \emph{Hadamard walk}. These notions are discussed in \cite{Quant, Konno}; see \cite{review} for a comprehensive survey.

Our main (new) result states that the probability to find an electron at a lattice point is nonzero if there is at least one checker path from the origin to that point (Theorem~\ref{non-zero}), answering a question by A.~Ustinov. Also we present several results on the average velocity of the electron ($\S$\ref{velocity_section}). We prove that the expectation of the average electron velocity equals the time-average of the expectation of the instantaneous velocity (Proposition~\ref{velocity}), answering a question by D.~Treschev. 
Also we compute the limit value of the average electron velocity when time tends to infinity (Theorem~\ref{limit_velocity}). 
This result has been proven in \cite[Theorem~1]{weak_limit} (see also an exposition in \cite[$\S$12.2]{main}), but we present a short elementary proof. In addition, we state a lot of new identities related to the model, which were found in numeric experiments (See $\S$\ref{identities}). A few of these identities are proven and the rest are interesting open problems.

\section{Definition and examples}
\label{def}

In this section we present the definition and physical interpretation of Feynman checkers from \cite{main}.

\begin{definition}[{\cite[Definition~2]{main}}] \label{def-mass}
Fix $\varepsilon>0$ and $m\geqslant 0$ called \emph{lattice step} and \emph{particle mass} respectively. Consider the lattice $\zz=\{\,(x,t):x/\varepsilon,t/\varepsilon\in\mathbb{Z}\,\}$. The elements of $\zz$ are called \emph{lattice points}. A \emph{checker path} is a finite sequence of points of $\zz$ such that the vector from
each point (except the last one) to the next one equals either $(\varepsilon,\varepsilon)$ or $(-\varepsilon,\varepsilon)$. A \emph{turn} is a point of the path
(not the first and not the last one) such that the vectors from the point to the next and to the previous ones are orthogonal. For $(x, t) \in \varepsilon\mathbb{Z}^2$, where $t>0$, denote  
$${a}(x,t,m,\varepsilon):=(1+m^2\varepsilon^2)^{(1-t/\varepsilon)/2}\,i\,\sum_s (-im\varepsilon)^{\mathrm{turns}(s)},$$
where the sum is over all checker paths $s$ from $(0,0)$ to $(x,t)$ with the first step to $(\varepsilon,\varepsilon)$, and $\mathrm{turns}(s)$ is the number of turns in $s$. Denote $$P(x,t,m,\varepsilon):=|{a}(x,t,m,\varepsilon)|^2.$$
Denote by $a_1(x,t,m,\varepsilon)$ and $a_2(x,t,m,\varepsilon)$ the real and the imaginary part of $a(x,t,m,\varepsilon)$ respectively.
\end{definition} 


\begin{remark}[Physical interpretation of the model, {\cite{main}}]
We use the natural system of units, where both the speed of light and the Plank constant equal $1$. The $t-$ and $x-$coordinates are interpreted as time and position of the particle of spin 1/2 and mass $m$ respectively. Thus any checker path is interpreted as motion of a particle in 1D space with the speed of light (with change of direction). The line $x=t$ represents motion in one direction with the speed of light. In what follows, we consider the motion of an electron (so that $m$ is the mass of an electron). The number $P(x,t,m,\varepsilon)$ is called the \emph{probability to find an electron at the lattice point $(x,t)$, if the electron was emitted from the point $(0,0)$}. Such terminology is confirmed by the fact that all the numbers $P(x, t,m,\varepsilon)$ on one horizontal sum up to $1$ (see Proposition~\ref{prob}). Figure~\ref{graph} shows $P(x,1000,1,1)$ for $|x| \leqslant 1000$. Note that if $x$ is greater than $750$, then the probability is very small but still non-zero (see Theorem~\ref{non-zero}).
\end{remark}

\begin{figure}[h]
	\centering
	\includegraphics[width=0.64\linewidth]{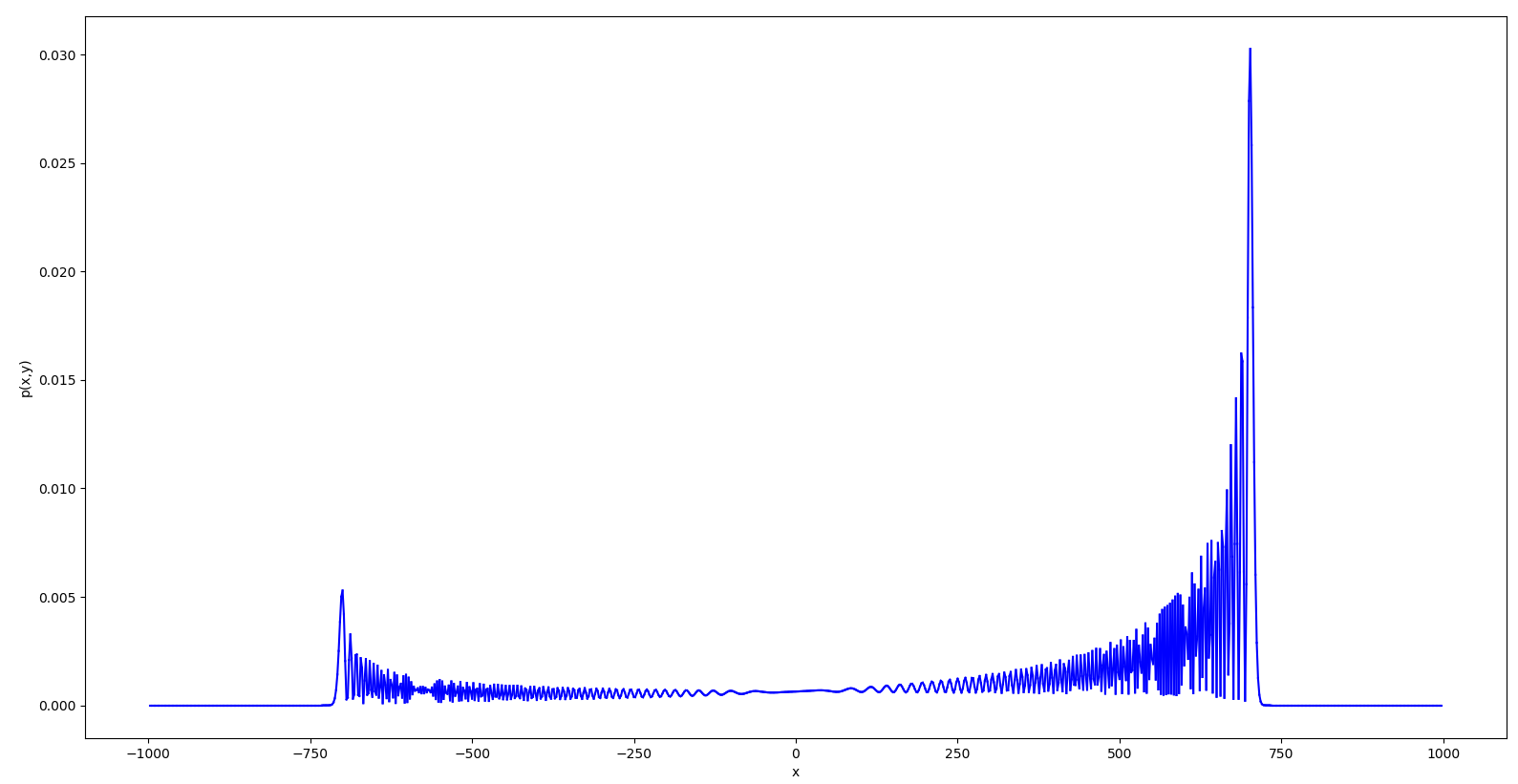}
	\caption{The graph of $P(x,1000,1,1)$ (A.~Daniyarkhodzhaev--F.~Kuyanov). Note that for $|x| \geqslant 750$ the probability is very small (but still non-zero).}
	\label{graph}
\end{figure}

\begin{example}
Let us compute $a(0,4\varepsilon,m,\varepsilon)$. Figure~\ref{example1} shows all three checker paths from $(0,0)$ to $(0, 4\varepsilon)$ starting with an upward-right move.  Thus by definition $a(0,4\varepsilon,m,\varepsilon) = - \frac{m^2 \varepsilon^2}{(1+m^2\varepsilon^2)^{3/2}} i$.
\end{example}

\begin{figure}[H]
\begin{center}
\includegraphics[width=0.17\textwidth]{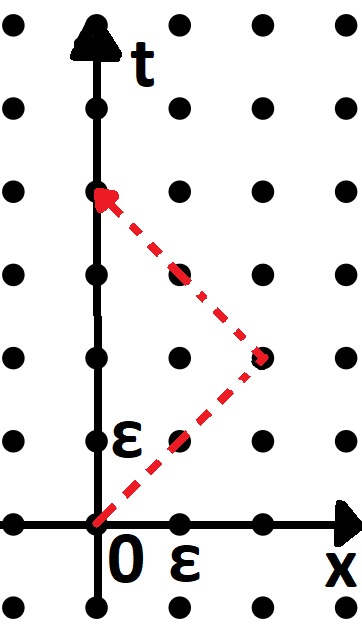}
\hfill
\includegraphics[width=0.17\textwidth]{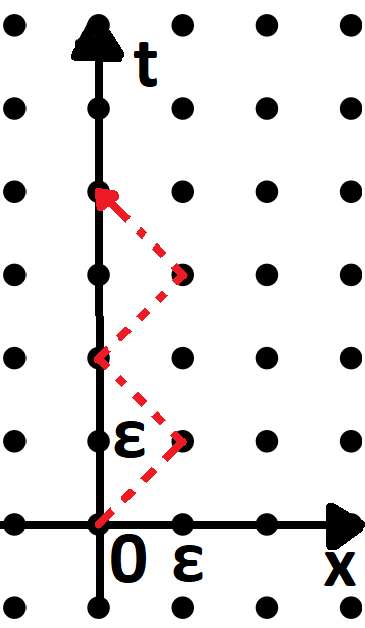} 
\hfill
\includegraphics[width=0.17\textwidth]{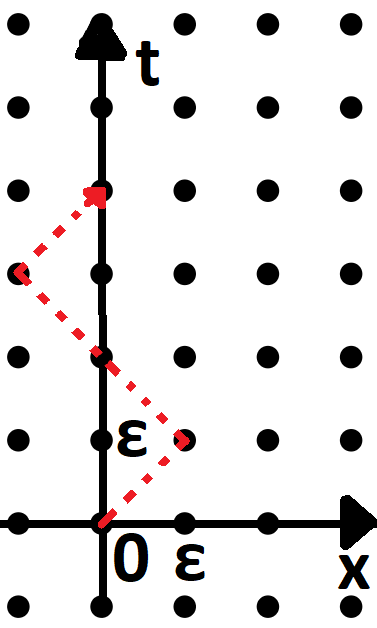} 
\caption{Checker paths contributing to $a(0,4\varepsilon,m,\varepsilon)$}
\label{example1}
\end{center}
\end{figure}

\begin{example} \label{edge}
It is easy to show that for each $x\in\zp$ we have: \\
$1) \: a(x,x,m,\varepsilon) = (1+m^2\varepsilon^2)^{(1-x/\varepsilon)/2} \, i ;$\\
$2) \: a(-x,x+2\varepsilon,m,\varepsilon) = m \varepsilon \, (1+m^2\varepsilon^2)^{-(1+x/\varepsilon)/2}.$
\end{example}

Let us  present several tables that show $a(x,t,m,\varepsilon)$ and $P(x,t,m,\varepsilon)$ for small $x$ and $t$. In Table~\ref{a-table}, the number in a cell $(x,t)$ is $a(x,t,m,\varepsilon)$, and an empty cell means that $a(x,t,m,\varepsilon) = 0$. Analogously, in Table~\ref{P-table}, the number in a cell $(x,t)$ is $P(x,t,m,\varepsilon)$, and an empty cell means that $P(x,t,m,\varepsilon) = 0$. Note that for fixed $t$ the sum of the  probabilities equals $1$.

\begin{table}[H]
\begin{tabular}{|c|c|c|c|c|c|c|c|}
\hline
$4\varepsilon$&$\frac{m \varepsilon}{(1+m^2\varepsilon^2)^{3/2}}$&&$\frac{(m \varepsilon- m^3 \varepsilon^3) - m^2 \varepsilon^2 i}{(1+m^2\varepsilon^2)^{3/2}}$&&$\frac{m \varepsilon - 2 m^2 \varepsilon^2 i}{(1+m^2\varepsilon^2)^{3/2}}$&&$\frac{1}{(1+m^2\varepsilon^2)^{3/2}} i$\\
\hline
$3\varepsilon$&&$\frac{m \varepsilon}{1+m^2\varepsilon^2}$&&$\frac{m \varepsilon - m^2 \varepsilon^2 i}{1+m^2\varepsilon^2}$&&$\frac{m^2 \varepsilon^2}{(1+m^2\varepsilon^2)} i$&\\
\hline
$2\varepsilon$&&&$\frac{m\varepsilon}{\sqrt{1+m^2\varepsilon^2}}$&&$\frac{1}{\sqrt{1+m^2\varepsilon^2}} i$&&\\
\hline
$\varepsilon$&&&&$i$&&&\\
\hline
\diagbox[dir=SW,height=21pt]{$t$}{$x$}&$-2\varepsilon$&$-\varepsilon$&$0$&$\varepsilon$&$2\varepsilon$&$3\varepsilon$&$4\varepsilon$ \\
\hline
\end{tabular}
\caption{$a(x,t,m,\varepsilon)$ for small $x$ and $t$.}
\label{a-table}
\end{table}

\begin{table}[H]
\begin{tabular}{|c|c|c|c|c|c|c|c|}
\hline
$4\varepsilon$&$\frac{m^2 \varepsilon^2}{(1+m^2\varepsilon^2)^3}$&&$\frac{m^2 \varepsilon^2 (1 - m^2 \varepsilon^2 + m^4 \varepsilon^4)}{(1+m^2\varepsilon^2)^3}$&&$\frac{m^2 \varepsilon^2 (1 + 4 m^2 \varepsilon^2)}{(1+m^2\varepsilon^2)^3}$&&$\frac{1}{(1+m^2\varepsilon^2)^3}$\\
\hline
$3\varepsilon$&&$\frac{m^2 \varepsilon^2}{(1+m^2\varepsilon^2)^2}$&&$\frac{m^2 \varepsilon^2}{1+m^2\varepsilon^2}$&&$\frac{1}{(1+m^2\varepsilon^2)^2}$&\\
\hline
$2\varepsilon$&&&$\frac{m^2\varepsilon^2}{1+m^2\varepsilon^2}$&&$\frac{1}{1+m^2\varepsilon^2}$&&\\
\hline
$\varepsilon$&&&&$1$&&&\\
\hline
\diagbox[dir=SW,height=21pt]{$t$}{$x$}&$-2\varepsilon$&$-\varepsilon$&$0$&$\varepsilon$&$2\varepsilon$&$3\varepsilon$&$4\varepsilon$ \\
\hline
\end{tabular}
\caption{$P(x,t,m,\varepsilon)$ for small $x$ and $t$.}
\label{P-table}
\end{table}

\section{Known results}
\label{preliminaries}

In this section, we state some properties of the model, Propositions~1-5 being folklore. They (and Propositions~6-7) are proved and discussed in \cite{main} and \cite{Bogdanov}. 

\begin{Symbol}
In what follows, we use the following notation:
\begin{flalign*}
&\zz=\{\,(x,t):x/\varepsilon,t/\varepsilon\in\mathbb{Z}\,\}; \\
&\z=\{\,t : t/\varepsilon\in\Z\,\}; \\
&\zp=\{\,t>0 : t/\varepsilon\in\Z\,\}; \\
&\Z_+=\{\,t>0 : t\in\Z\,\}. &&
\end{flalign*}
\end{Symbol}

\begin{proposition}[Dirac equation; {\cite[Proposition~5]{main}}]\ \label{Dirac} For each $(x,t)\in\zz$, where $t\geqslant 2\varepsilon$, we have:
\begin{align*}
1) \: a_1(x,t,m,\varepsilon) &= \frac{1}{\sqrt{1+m^2\varepsilon^2}}(a_1(x+\varepsilon,t-\varepsilon,m,\varepsilon) + m \varepsilon \, a_2(x+\varepsilon,t-\varepsilon,m,\varepsilon)); \\
2) \: a_2(x,t,m,\varepsilon) &= \frac{1}{\sqrt{1+m^2\varepsilon^2}}(a_2(x-\varepsilon,t-\varepsilon,m,\varepsilon) - m \varepsilon \, a_1(x-\varepsilon,t-\varepsilon,m,\varepsilon)).
\end{align*}
\end{proposition}

\begin{proposition}[{\cite[Lemma~1]{main}}] \ \label{formulas-down} For each $(x,t) \in \zz$ , where $t \geqslant 2\varepsilon$, we have:
\begin{align*}
1) \: a_1(x,t-\varepsilon,m,\varepsilon) &= \frac{1}{\sqrt{1+m^2\varepsilon^2}} (a_1(x-\varepsilon,t,m,\varepsilon) - m \varepsilon \, a_2(x+\varepsilon,t,m,\varepsilon)) ;\\
2) \: a_2(x,t-\varepsilon,m,\varepsilon) &= \frac{1}{\sqrt{1+m^2\varepsilon^2}} (a_2(x+\varepsilon,t,m,\varepsilon) + m \varepsilon \, a_1(x-\varepsilon,t,m,\varepsilon)).
\end{align*}
\end{proposition}

\begin{proposition}[Formulae for $a_1(x,t,m,\varepsilon)$ and $a_2(x,t,m,\varepsilon)$; {\cite[Proposition~11]{main}}] For all $(x,t) \in\zz$ such that $t>|x|$ and $(x+t)/\varepsilon$ is even we have:
\label{formula}
$$a_1(x,t,m, \varepsilon) = (1+ m^2 \varepsilon^2)^{(1-t/\varepsilon)/2}\sum_{r=0}^{(t-|x|)/2\varepsilon} (-1)^r \binom{(t+x-2\varepsilon)/2\varepsilon}{r} \binom{(t-x-2\varepsilon)/2\varepsilon}{r} (m \varepsilon)^{2r+1};$$
$$a_2(x,t,m, \varepsilon) = (1+ m^2 \varepsilon^2)^{(1-t/\varepsilon)/2}\sum_{r=1}^{(t-|x|)/2\varepsilon} (-1)^r \binom{(t+x-2\varepsilon)/2\varepsilon}{r} \binom{(t-x-2\varepsilon)/2\varepsilon}{r-1} (m \varepsilon)^{2r}.$$
\end{proposition}

\begin{proposition}[Probability conservation law; {\cite[Proposition~6]{main}}] \label{prob} 
For each $t\!\in\!\zp$ we get 
$$\sum\limits_{x\in\z}P(x,t,m, \varepsilon)=1.$$
\end{proposition}

\begin{proposition}[Symmetry; {\cite[Proposition~8]{main}}] For all $(x, t) \in \zz$ with  $t>0$ we have: \label{symmetry} \\
    1) $a_1(x,t,m,\varepsilon) = a_1(-x,t,m,\varepsilon)$;\\
    2) $(t-x) a_2(x,t,m,\varepsilon) = (t+x-2\varepsilon) a_2(2\varepsilon-x,t,m,\varepsilon)$.
\end{proposition}

\begin{proposition}[{\cite[Theorem~5, its proof, and Remark~1]{main}}]\label{p-right-prob}
If $t\in \z_+$ and $m \varepsilon = 1$, then 
$$\sum_{x\in\z}a_{1}(x,t,m,\varepsilon)^2
 = \frac{1}{2} \sum\limits_{k=0}^{\floor{t/2\varepsilon}-1} \frac{1}{(-4)^k}\binom{2k}{k}
=\frac{1}{2\sqrt{2}}+\mathrm{O}\left(\sqrt{\frac{\varepsilon}{t}}\right).$$
\end{proposition}

Hereafter notation $f(t,m,\varepsilon) = \mathrm{O}(g(t,m,\varepsilon))$ means there is a constant $C$ (not depending on $t,m,\varepsilon$) such that for each $t,m,\varepsilon$ satisfying the assumptions of the proposition we have $|f(t,m,\varepsilon)| \leqslant C g(t,m,\varepsilon)$.

The following genealization of Proposition~\ref{p-right-prob} was conjectured by I.~Gaidai-Turlov, T.~Kovalev, A.~Lvov in 2019, and proved by I.~Bogdanov in 2020.

\begin{proposition}[{\cite[Theorem~2]{Bogdanov}}] \label{p-right-prob-gen}
If $0 \leqslant m \varepsilon \leqslant 1$, then $$\underset{t\in \z_+}{\lim_{t \to +\infty}}\sum_{x\in\z}a_{1}(x,t,m,\varepsilon)^2 = \frac{m \varepsilon}{2 \sqrt{1+m^2 \varepsilon^2}}.$$
\end{proposition}

\section{Main result: the probability to find an electron vanishes nowhere inside the light cone}
\label{new results}

The goal of this section is to prove the following theorem.

\begin{theorem}
\label{non-zero}
    For each $m>0$ and a point $(x,t) \in \zz$ such that $(x + t)/\varepsilon$ is even and $t>|x|$ we have $P(x,t,m,\varepsilon) \neq 0$.
\end{theorem}

\begin{remark}
Thus $P(x,t,m,\varepsilon) \neq 0$ if and only if $(x + t)/\varepsilon$ is even and either $t>|x|$ or $t=x>0$. In other words, $P(x,t,m,\varepsilon) \neq 0$ if and only if there exists at least one checker path from $(0,0)$ to $(x,t)$.
\end{remark}

\begin{remark}
Note that the assertion of the theorem is not obvious at all. It follows neither from Definition~\ref{def-mass} nor from the explicit formulae from Proposition~\ref{formula}. Figure~\ref{graph} shows that for some pairs $(x,t)$ the probability can be a very small number. But Proposition~\ref{symmetry} and a simple trick helps us to prove the theorem.
\end{remark}

\begin{proof}[Proof of Theorem~\ref{non-zero}]
Denote $M = \{(x,t) \in \zz \: : \: (x + t)/\varepsilon$ is even, $t>|x|, P(x,t,m,\varepsilon) = 0\}$. If $M=\emptyset$, then there is nothing to prove. Assume that $M \neq \emptyset$. Among the points of $M$, select the one with the minimal $t-$coordinate (if there are several such points, select any of them). Denote by $(x_0,t_0)$ the selected point. By Example~\ref{edge}, for all $t\in\zp$ we have $P(-t+2\varepsilon,t) = m^2 \varepsilon^2 (1+m^2\varepsilon^2)^{(1-t/\varepsilon)} \neq 0$, and thus $x_0 \neq -t_0+2 \varepsilon$. Subsequently, by Proposition~\ref{symmetry} it follows that $a_1(-x_0,t_0,m,\varepsilon) = a_1(x_0,t_0,m,\varepsilon) = 0$ and $a_2(2\varepsilon-x_0,t_0,m,\varepsilon) = (t_0-x_0) \frac{a_2(x_0,t_0,m,\varepsilon)}{t_0+x_0-2\varepsilon} = 0$. See Figure~\ref{zeros}. By Proposition~\ref{formulas-down} we get 
\begin{align*}
a_1(-x_0+\varepsilon,t_0-\varepsilon, m, \varepsilon) &= \frac{a_1(-x_0,t_0,m,\varepsilon) - m \varepsilon \, a_2(-x_0+2\varepsilon,t_0,m,\varepsilon)}{\sqrt{1+m^2\varepsilon^2}} = 0 \text{ and}\\
a_2(-x_0+\varepsilon,t_0-\varepsilon, m, \varepsilon) &= \frac{a_2(-x_0+2\varepsilon,t_0,m,\varepsilon) + m \varepsilon \, a_1(-x_0,t_0,m,\varepsilon)}{\sqrt{1+m^2\varepsilon^2}} = 0.
\end{align*}
Thus $P(-x_0+\varepsilon, t_0-\varepsilon,m,\varepsilon) = 0$. This contradicts to the minimality of $t_0$, because $(x_0-\varepsilon+t_0-\varepsilon)/\varepsilon$ is even and $t_0-\varepsilon > |x_0 - \varepsilon|$ by the condition $x_0 \neq -t_0 + 2\varepsilon$ above.
\end{proof}

\begin{figure}[h]
	\centering
	\begin{minipage}{0.38\textwidth}
		\centering
		\includegraphics[width=1\linewidth]{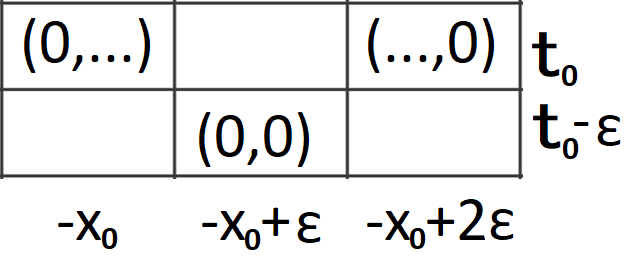}
		\caption{The pair in a cell $(x,t)$ is $(a_1(x,t), a_2(x,t))$}
		\label{zeros}
	\end{minipage}
\end{figure}

\section{On the electron velocity}
\label{velocity_section}

In this section we prove Proposition~\ref{velocity} and Theorem~\ref{limit_velocity} stated below. The former  answers a question by D.~Treschev. For the statements, we need the following definitions.

\subsection{The expectation of the average electron velocity equals the time-average of the expectation of the instantaneous electron velocity}
\label{velocity_electron1}

It goes almost without saying to consider the electron as being in one of the two states depending on the last-move direction: right-moving or left-moving (or just `right' or `left' for brevity). Actually, these two states are exactly (1+1)-dimensional analogue of chirality states for a spin 1/2 particle. See, for example, a discussion on this topic in \cite{main}. Let us give a few new definitions.

\begin{definition}[{\cite{main}}]
The \emph{probability $P_{t,m,\varepsilon}(x,+)$ to find a right electron at the lattice point $(x,t)\in\zz$, if the right electron was emitted from the point $(0,0)$}, is the length square of the vector ${a}(x,t,m,\varepsilon,+):=(1+m^2\varepsilon^2)^{(1-t/\varepsilon)/2}\,i\,\sum_s (-im\varepsilon)^{\mathrm{turns}(s)}$, where the sum is over only those checker paths from $(0,0)$ to $(x,t)$ that both start and finish with an upwards-right move, and $\mathrm{turns}(s)$ is the number of turns in $s$. \\
The \emph{probability $P_{t,m,\varepsilon}(x,-)$ to find a left electron} is defined analogously, only the sum is taken over checker paths that start with an upwards-right move but finish with an upwards-left move.
\end{definition}

\begin{remark}[{\cite{main}}]
Clearly, the above probabilities equal $a_2(x,t,m,\varepsilon)^2$ and $a_1(x,t,m,\varepsilon)^2$ respectively, because the last move is directed upwards-right if and only if the number of turns is even. By Proposition~\ref{prob}, $P_{t,m,\varepsilon}(x,\pm)$ is indeed a probability measure on the set $\{(x,\pm) : x\in\z\}$.
\end{remark}


\begin{definition}
\label{vel}
\noindent The \emph{average electron velocity} is the random variable on the set $\{(x,\pm) : x \in \z\}$ given by
$$v_t(x,\pm) = v_t(x)=x/t.$$
\noindent The \emph{instantaneous electron velocity} is the random variable on the same set given by
$$u_t(x,\sigma) = \begin{cases}
+1 \text{, if }\sigma = +,\\
-1 \text{, if }\sigma = -.
\end{cases}$$
\end{definition}

\begin{remark}
The random variables in the above definition depend on $t, m, \varepsilon$. To emphasize that, we use the notation $\ee_{t,m,\varepsilon}$ for their expectation.
\end{remark}

\begin{proposition} \label{velocity} 
The expectation of the average electron velocity equals the time-average of the expectation of the instantaneous electron velocity, i.e., for each $T\in\zp$ we have 
$$\ee_{T,m,\varepsilon}\left(v_T\right) = \frac{\varepsilon}{T}\left[\sum\limits_{t=\varepsilon,\ldots,T}\ee_{t,m,\varepsilon}\left(u_t\right)\right].$$
\end{proposition}

\begin{remark}
This is not at all automatic because the model does not give any probability distribution on the set of checker paths. Moreover, notice that in the expression $\sum\limits_{t=\varepsilon,\ldots,T}\ee_{t,m,\varepsilon}\left(u_t\right)$ we take expectations in $T/\varepsilon$ different probability spaces.
\end{remark}

\begin{proof}[Proof of Proposition~\ref{velocity}]
If $T=\varepsilon$, then there is nothing to prove. Thus we assume that $T\geqslant 2 \varepsilon$.

By Proposition~\ref{Dirac} for each $t>\varepsilon$ we have:
\begin{align*}
&\sum\limits_{x\in\z}a_2(x,t,m,\varepsilon)^2-\sum\limits_{x\in\z}a_1(x,t,m, \varepsilon)^2 = \\
& = \frac{1}{1 + m^2 \varepsilon^2} \bigg[ \sum\limits_{x\in\z}a_2(x,t-\varepsilon,m,\varepsilon)^2+ m^2 \varepsilon^2 \sum\limits_{x\in\z}a_1(x,t-\varepsilon,m, \varepsilon)^2 - \\
& - 2 m \varepsilon \sum\limits_{x\in\z} a_1(x,t-\varepsilon,m, \varepsilon) a_2(x,t-\varepsilon,m, \varepsilon) \bigg] - \\
& - \frac{1}{1 + m^2 \varepsilon^2} \bigg[ \sum\limits_{x\in\z}a_1(x,t-\varepsilon,m,\varepsilon)^2+ m^2 \varepsilon^2 \sum\limits_{x\in\z}a_2(x,t-\varepsilon,m, \varepsilon)^2 + \\
& + 2 m \varepsilon \sum\limits_{x\in\z} a_1(x,t-\varepsilon,m, \varepsilon) a_2(x,t-\varepsilon,m, \varepsilon) \bigg] = \\
& = \frac{1 - m^2 \varepsilon^2}{1 + m^2 \varepsilon^2} \left(\sum\limits_{x\in\z}a_2(x,t-\varepsilon,m,\varepsilon)^2 - \sum\limits_{x\in\z}a_1(x,t-\varepsilon,m,\varepsilon)^2\right) - \\
&- \frac{4 m \varepsilon}{1 + m^2 \varepsilon^2} \sum\limits_{x\in\z} a_1(x,t-\varepsilon,m, \varepsilon) a_2(x,t-\varepsilon,m, \varepsilon).
\end{align*}

This implies that
\begin{align}
&\varepsilon \sum\limits_{t=\varepsilon,\ldots,T}\ee_{t,m,\varepsilon}\left(u_t\right) = 
\varepsilon \sum\limits_{t=\varepsilon,\ldots,T}\left(\sum\limits_{x\in\z}a_2(x,t,m, \varepsilon)^2-\sum\limits_{x\in\z}a_1(x,t,m, \varepsilon)^2\right) = \nonumber \\
& = \varepsilon + \varepsilon \frac{1 - m^2 \varepsilon^2}{1 + m^2 \varepsilon^2} \sum\limits_{t=2 \varepsilon,\ldots,T}\left(\sum\limits_{x\in\z}a_2(x,t-\varepsilon,m,\varepsilon)^2 - \sum\limits_{x\in\z}a_1(x,t-\varepsilon,m,\varepsilon)^2\right) - \nonumber \\
&- \varepsilon \frac{4 m \varepsilon}{1 + m^2 \varepsilon^2} \sum\limits_{t=2 \varepsilon,\ldots,T}\sum\limits_{x\in\z} a_1(x,t-\varepsilon,m, \varepsilon) a_2(x,t-\varepsilon,m, \varepsilon), \label{ur1}
\end{align}
where we use the obvious equality $\varepsilon \left( \sum\limits_{x\in\z}a_2(x,\varepsilon,m,\varepsilon)^2 - \sum\limits_{x\in\z} a_1(x,\varepsilon,m,\varepsilon)^2 \right) = \varepsilon$. On the other hand, by Proposition~\ref{Dirac} we have:
\begin{align}
T \cdot \ee_{T,m,\varepsilon} \left(v_T\right) &= \sum\limits_{x\in\z}x a_1(x,T,m, \varepsilon)^2 + \sum\limits_{x\in\z} x a_2(x,T,m, \varepsilon)^2 = \nonumber \\
& =  \sum\limits_{x\in\z}(x-\varepsilon) a_1(x-\varepsilon,T,m, \varepsilon)^2 + \sum\limits_{x\in\z} (x+\varepsilon) a_2(x+\varepsilon,T,m, \varepsilon)^2 = \nonumber \\
& = \frac{1}{1+m^2\varepsilon^2} \sum\limits_{x\in\z}\bigg[(x-\varepsilon) a_1(x,T-\varepsilon,m, \varepsilon)^2 + (x-\varepsilon) m^2 \varepsilon^2 \, a_2(x,T-\varepsilon,m, \varepsilon)^2 + \nonumber \\
& + 2(x-\varepsilon) m \varepsilon \, a_1(x,T-\varepsilon,m, \varepsilon) a_2(x,T-\varepsilon,m, \varepsilon) \bigg] + \nonumber \\
& + \frac{1}{1+m^2\varepsilon^2} \sum\limits_{x\in\z}\bigg[(x+\varepsilon) a_2(x,T-\varepsilon,m, \varepsilon)^2 + (x+\varepsilon) m^2 \varepsilon^2 \, a_1(x,T-\varepsilon,m, \varepsilon)^2 - \nonumber \\
& - 2(x+\varepsilon) m \varepsilon \, a_1(x,T-\varepsilon,m, \varepsilon) a_2(x,T-\varepsilon,m, \varepsilon)\bigg] = \nonumber \\
&= \sum\limits_{x\in\z}x a_1(x,T-\varepsilon,m, \varepsilon)^2 + \sum\limits_{x\in\z} x a_2(x,T-\varepsilon,m, \varepsilon)^2 + \nonumber \\
& + \varepsilon \frac{1 - m^2 \varepsilon^2}{1 + m^2 \varepsilon^2} \left(\sum\limits_{x\in\z}a_2(x,T-\varepsilon,m,\varepsilon)^2 - \sum\limits_{x\in\z}a_1(x,T-\varepsilon,m,\varepsilon)^2\right) -  \nonumber \\
&- \frac{4 m \varepsilon^2}{1+m^2\varepsilon^2}  \sum\limits_{x\in\z} a_1(x,T-\varepsilon,m, \varepsilon) a_2(x,T-\varepsilon,m, \varepsilon) = \nonumber \\
&= \varepsilon + \varepsilon \frac{1 - m^2 \varepsilon^2}{1 + m^2 \varepsilon^2} \sum\limits_{t=2 \varepsilon,\ldots,T}\left(\sum\limits_{x\in\z}a_2(x,t-\varepsilon,m,\varepsilon)^2 - \sum\limits_{x\in\z}a_1(x,t-\varepsilon,m,\varepsilon)^2\right) - \nonumber \\
& - \frac{4 m \varepsilon^2}{1+m^2\varepsilon^2} \sum\limits_{t=2\varepsilon,\ldots,T}\sum\limits_{x\in\z} a_1(x,t-\varepsilon,m, \varepsilon) a_2(x,t-\varepsilon,m, \varepsilon), \label{ur2}
\end{align} 
where the latter equality is obtained by repeating the same transformation $T-2\varepsilon$ times.

Comparing \eqref{ur1} and \eqref{ur2}, we get the required result.
\end{proof}

\subsection{Analogue of Proposition~\ref{velocity} for classical random walks}

To show that the statement of Proposition~\ref{velocity} is natural, let us give its analogue for classical random walks. Let us consider a flea that makes a random walk on the integer number line starting from $0$. If the flea is situated in the number $N$ at the time $t_0$, then at the time $t_0+1$ it is situated in the number $N+1$ with the probability $p$ or in the number $N-1$ with the probability $q = 1-p$. If we denote the probability to find the flea at $x\in\mathbb{N}$ at the time $t$ by $\widehat P(x,t)$, then $\widehat P(x,t) = p \widehat P(x-1,t-1) + q \widehat P(x+1,t-1)$.

\begin{definition} 
For a point $(x,t) \in \mathbb{Z}^2 \: (t \geqslant 0)$, we define the \emph{probability $\widehat P(x,t)$ to find the flea at the lattice point $(x,t)$, if the flea makes a random walk} by induction on $t$: \\
$\widehat P(x,0) = 
\begin{cases}
1, &\text{if }x=0,\\
0, &\text{if }x \neq 0.
\end{cases}$ \\
$\widehat P(x,t) = p \widehat P(x-1,t-1) + q \widehat P(x+1,t-1)$ for $t\geqslant1$.
\end{definition}

\begin{remark}
As in Definition~\ref{def-mass}, $x$ is interpreted as position and $t$ is interpreted as time. 
\end{remark}

\begin{remark}
It is easy to prove that for fixed $t$ we have $\sum_{x\in\mathbb{Z}} \widehat P(x,t) = 1$.
\end{remark}

The average flea velocity $\widehat v_t$ and the instantaneous flea velocity $\widehat u_t$ are defined literally as in Definition~\ref{vel}. The random variables $\widehat v_t$ and $\widehat u_t$ depend on $t\in\Z_+$; we write $\ee_{t}$ for their expectations.

 The following easy well-known proposition is an analogue of Proposition~\ref{velocity}.

\begin{proposition}[An analogue of Proposition~\ref{velocity} for classical random walks]
The expectation of the average flea velocity equals the time-average of the expectation of the instantaneous flea velocity, i.e., for each $T\in\Z_+$ we have 
$$\ee_{T}\left(\widehat v_T\right) = \frac{1}{T}\left[\sum\limits_{t=1}^T \ee_{t}\left(\widehat u_t\right)\right].$$
\end{proposition}

\begin{proof}
Clearly, for each $T\in\Z_+$ we have $\ee_T(\widehat u_T) = p-q$. Let us prove by induction on $T\in\Z_+$ that for each $T$ we have $\ee_T (\widehat v_T) = p-q$. The base is obvious. To perform the induction step, suppose that for $T\in\Z_+$ we have $\ee_T (\widehat v_T) = p-q$. Then
\begin{align*}
\ee_{T+1} (\widehat v_{T+1}) &= \frac{1}{T+1} \sum\limits_{x\in\mathbb{Z}} x \widehat P(x,T+1) = \frac{1}{T+1} \sum\limits_{x\in\mathbb{Z}} x p \widehat P(x-1,T) + \frac{1}{T+1} \sum\limits_{x\in\mathbb{Z}} x q \widehat P(x+1,T) \\
& = \frac{1}{T+1} \sum\limits_{x\in\mathbb{Z}} (x+1) p \widehat P(x,T) + \frac{1}{T+1} \sum\limits_{x\in\mathbb{Z}} (x-1) q \widehat P(x,T) =  \\
& \frac{1}{T+1} \sum\limits_{x\in\mathbb{Z}} \left( x p \widehat P(x,T) + x q \widehat P(x,T) \right)+  \frac{1}{T+1} \sum\limits_{x\in\mathbb{Z}} (p-q) \widehat P(x,T) = \\
& \overset{(a)}{=} \frac{1}{T+1}\ee_{T} (\widehat v_{T}) + \frac{p-q}{T+1} \overset{(b)}{=} \frac{T (p-q)}{T+1} + \frac{p-q}{T+1} = p-q,
\end{align*}
where (a) follows from the identity $\sum_{x\in\mathbb{Z}} \widehat P(x,T) = 1$ for fixed $T$, and (b) follows from the inductive hypothesis. Thus $\ee_T (\widehat v_T) = p-q$ for all $T\in\Z_+$.
\end{proof}

\subsection{The limit value of the average electron velocity}
The following theorem gives us the limit value of the average electron velocity when time tends to infinity. This theorem and Proposition~\ref{limit_velocity_with_O} are corollaries of a more general result \cite[(18)]{weak_limit}, but we present short elementary proofs.

\begin{theorem} \label{limit_velocity}
If $0 \leqslant m \varepsilon \leqslant 1$, then we have
$$\underset{T\in\z_+}{\lim_{T \to +\infty}}\ee_{T,m,\varepsilon}\left(v_T\right) = 1 - \frac{m \varepsilon}{\sqrt{1+m^2 \varepsilon^2}}.$$
\end{theorem}

\begin{proof}
By Propositions~\ref{velocity} and \ref{prob} we have
\begin{align*}
\ee_{T,m,\varepsilon}\left(v_T\right) &= \frac{\varepsilon}{T}\left(\sum\limits_{t=\varepsilon,\ldots,T}\ee_{t,m,\varepsilon}\left(u_t\right)\right) = \frac{\varepsilon}{T} \sum\limits_{t=\varepsilon,\ldots,T}\left(\sum\limits_{x\in\z}a_2(x,t,m, \varepsilon)^2-\sum\limits_{x\in\z}a_1(x,t,m, \varepsilon)^2\right) = \\
& = \frac{\varepsilon}{T} \sum\limits_{t=\varepsilon,\ldots,T}\left(1 - 2 \sum\limits_{x\in\z}a_1(x,t,m, \varepsilon)^2\right) = \frac{\varepsilon}{T} \frac{T}{\varepsilon} -\frac{2 \varepsilon}{T} \sum\limits_{t=\varepsilon,\ldots,T}\sum\limits_{x\in\z}a_1(x,t,m, \varepsilon)^2 = \\
& = 1 - \frac{m \varepsilon}{\sqrt{1+m^2 \varepsilon^2}} - \frac{2 \varepsilon}{T} \sum\limits_{t=\varepsilon,\ldots,T}\left(\sum\limits_{x\in\z}a_1(x,t,m, \varepsilon)^2 - \frac{m \varepsilon}{2 \sqrt{1+m^2 \varepsilon^2}} \right) =: 1 - \frac{m \varepsilon}{\sqrt{1+m^2 \varepsilon^2}} - \Delta(T,m,\varepsilon).
\end{align*}

Thus it remains to prove that $\Delta(T,m,\varepsilon) \to 0$ as $T \to +\infty$.

Take $\delta>0$. By Proposition~\ref{p-right-prob-gen} there exists $t_0 \in \zp$ such that for any $t>t_0$ we have
\begin{equation*}
\label{eq1}
\left|\sum\limits_{x\in\z}a_1(x,t,m, \varepsilon)^2 - \frac{m \varepsilon}{2 \sqrt{1+m^2 \varepsilon^2}}\right| < \frac{\delta}{3}.
\end{equation*}

Then by Proposition~\ref{prob} we have
\begin{align*}
& |\Delta(T,m,\varepsilon)| \leqslant \frac{2 \varepsilon}{T} \sum\limits_{t=\varepsilon,\ldots,t_0}\sum\limits_{x\in\z}a_1(x,t,m, \varepsilon)^2 + \frac{2 \varepsilon}{T} \sum\limits_{t=\varepsilon,\ldots,t_0} \frac{m \varepsilon}{2 \sqrt{1+m^2 \varepsilon^2}}  + \\
& + \frac{2 \varepsilon}{T} \sum\limits_{t=t_0 + \varepsilon,\ldots,T}\left|\sum\limits_{x\in\z}a_1(x,t,m, \varepsilon)^2 - \frac{m \varepsilon}{2 \sqrt{1+m^2 \varepsilon^2}} \right| \leqslant \\
& \leqslant \frac{2 \varepsilon}{T} \frac{t_0}{\varepsilon} + \frac{2 \varepsilon}{T} \frac{t_0}{\varepsilon} + \frac{2 \varepsilon}{T} \frac{T}{\varepsilon} \frac{\delta}{3} \leqslant \frac{4t_0}{T} + \frac{2 \delta}{3} < \delta \qquad \text{ for } T > \frac{12 t_0}{\delta}.\\
\end{align*}

Since $\delta$ is arbitrary, it follows that $\displaystyle{\lim_{T \to +\infty} \Delta(T,m, \varepsilon) = 0}$ and the theorem follows.
\end{proof}


Theorem~\ref{limit_velocity} provides the limit value, but tells nothing about the convergence rate. But the following proposition provides both for the particular case $m \varepsilon = 1$.

\begin{proposition} \label{limit_velocity_with_O}
For each $\varepsilon>0$ an each $T\in\z_+$ we have 
$$\ee_{T,1/\varepsilon,\varepsilon}\left(v_T\right) = 1 - \frac{1}{\sqrt{2}} + \mathrm{O}\left(\frac{\varepsilon}{T}\right).$$
\end{proposition}

\begin{proof}
Analogously to the proof of Theorem~\ref{limit_velocity}, for each $T\in\z_+$ we have
$$\ee_{T,1/\varepsilon,\varepsilon}\left(v_T\right) = 1 - \frac{1}{\sqrt{2}} - \frac{2 \varepsilon}{T} \sum\limits_{t=\varepsilon,\ldots,T}\left(\sum\limits_{x\in\z}a_1(x,t,1/\varepsilon,\varepsilon)^2 - \frac{1}{2 \sqrt{2}} \right).$$

It remains to show that 
$$\sum\limits_{t=\varepsilon,\ldots,T}\left(\sum\limits_{x\in\z}a_1(x,t,1/\varepsilon,\varepsilon)^2 - \frac{1}{2 \sqrt{2}} \right) =: \sum\limits_{t=\varepsilon,\ldots,T} c_t = \mathrm O(1).$$
This follows from the following facts:
\begin{enumerate}[(i)]
	\item If $t = 2 n \varepsilon$, where $n\in\Z$, then $c_t = c_{t + \varepsilon}$;
	\item $c_t \to 0$ as $t \to \infty;$
	\item For each $t\in \z, t > \varepsilon$, the signs of $c_t$ and $c_{t + 2 \varepsilon}$ are opposite;
	\item For each $t\in\z, t > \varepsilon$, we have $|c_t| > |c_{t + 2\varepsilon}|$.
\end{enumerate}
Here (i) and (ii) follow directly from Proposition~\ref{p-right-prob}. Assertion (iii) follows from Proposition~\ref{p-right-prob} and the fact that $\frac{1}{4^k}\binom{2k}{k} > \frac{1}{4^{k+1}}\binom{2(k+1)}{k+1}$ for $k \in Z_+$. Also note that from Proposition~\ref{p-right-prob} and (iii) we know that $c_t$ is positive if $t=(4n + 2)\varepsilon$ for some $n\in\Z_+ \cup \{0\}$, and $c_t$ is negative if $t=4n\varepsilon$ for some $n\in\Z_+$.

Let us prove (iv). Fix any $t = (4 n + 2) \varepsilon$, where $n \in \Z_+$. By Proposition~\ref{p-right-prob} we have 
$$c_{t + 2 m \varepsilon} = -\frac{1}{2 \sqrt{2}} + \frac{1}{2} \sum\limits_{k=0}^{2 n+m} \frac{1}{(-4)^{k}} \binom{2k}{k} \; (\text{here }m \in \Z_+).$$
Note that 
$$c_{t + 6\varepsilon} + c_{t + 4\varepsilon} = c_{t + 2\varepsilon} + c_{t} + \frac{1}{2}\left(- \frac{1}{4^{2n+1}}\binom{4n+2}{2n+1} + \frac{2}{4^{2n+2}}\binom{4n+4}{2n+2} - \frac{1}{4^{2n+3}}\binom{4n+6}{2n+3}\right).$$
It easy to show that the latter expression in parentheses is negative for any $n\in\Z_+$. Thus
\begin{equation}
\label{q11}
c_{t + 2\varepsilon} + c_{t} > c_{t + 6\varepsilon} + c_{t + 4\varepsilon} > c_{t + 10\varepsilon} + c_{t + 8 \varepsilon} > \ldots
\end{equation}
To prove (iv), we want to show that $|c_t| > |c_{t+2\varepsilon}|$. Equivalently, we want to prove that $c_t + c_{t+2\varepsilon} > 0$. Suppose that this is not true. Then by \eqref{q11} the sequence $\{c_{t + (2 m +2) \varepsilon} + c_{t + 2 m \varepsilon}\}_{m\in\Z_+}$ does not have $0$ as its limit. On the other hand, by (ii) this sequence should have $0$ as its limit. This is a contradiction. One can obtain a similar contradiction for the case when $t = 4n \varepsilon$.

Hence by (i)-(iv) we have $\left|\sum\limits_{t=\varepsilon,\ldots,T} c_t\right| \leqslant 2 c_{2\varepsilon}$ for each $T$, and thus $\sum\limits_{t=\varepsilon,\ldots,T} c_t = \mathrm O(1)$.

\end{proof}

\section{Identities}
\label{identities}

In this section, we present several identities in Feynman checkers. Not all of them pretend to be new, but we could not find them in the literature. All the identities were first discovered in Wolfram Mathematica, sometimes with the help of the On-Line Encyclopedia of Integer Sequences \cite{oeis}. The identities in this subsection should be considered as just some combinatorial equalities, which have absolutely no physical meaning. In all the identities below we assume that $m>0$.

\subsection{Linear identities}

From Proposition~\ref{prob} we know that for each $t\!\in\!\zp$ we have $\sum\limits_{x\in\z} \left( a_1(x,t,m,\varepsilon)^2 + a_2(x,t,m,\varepsilon)^2 \right) = 1$. But what if we consider not the sum of squares, but the sum of $a_1$-s and $a_2$-s themselves? The following proposition gives the answer.

\begin{proposition} \label{id1}
For each $t\in\zp$ we have
$$\sum\limits_{x\in\z} a_1(x,t,m,\varepsilon) = \sin{\left(\frac{t-\varepsilon}{\varepsilon}\arctan{m \varepsilon}\right)} \text{ and } \sum\limits_{x\in\z} a_2(x,t,m,\varepsilon) = \cos{\left(\frac{t-\varepsilon}{\varepsilon}\arctan{m \varepsilon}\right)}.$$
\end{proposition}

\begin{proof}
We prove both formulae simultaneously by induction on $t$. The base $t=\varepsilon$ is obvious. To perform the induction step, suppose that for $t \in \zp$ the formulae from the statement hold. Then by Proposition~\ref{Dirac} and by the summation formulae for the sine and the cosine we have 
\begin{align*}
\sum\limits_{x\in\z} a_1(x,t+\varepsilon,m,\varepsilon) &= \frac{1}{\sqrt{1+m^2 \varepsilon^2}} \sum\limits_{x\in\z} a_1(x,t,m,\varepsilon) + \frac{m \varepsilon}{\sqrt{1+m^2 \varepsilon^2}}\sum\limits_{x\in\z} a_2(x,t,m,\varepsilon) = \\
&= \frac{\sin{\left((t/\varepsilon-1)\arctan{m \varepsilon}\right)}}{\sqrt{1+m^2 \varepsilon^2}} +  \frac{m \varepsilon \cos{\left((t/\varepsilon-1) \arctan{m \varepsilon}\right)}}{\sqrt{1+m^2 \varepsilon^2}} = \\
&=\left[\frac{\sin{\left((t/\varepsilon) \arctan{m \varepsilon}\right)}}{1+m^2 \varepsilon^2} - \frac{m \varepsilon \cos{\left((t/\varepsilon)\arctan{m \varepsilon}\right)}}{1+m^2 \varepsilon^2}\right] + \\
& + \left[\frac{m \varepsilon \cos{\left((t/\varepsilon)\arctan{m \varepsilon}\right)}}{1+m^2 \varepsilon^2} + \frac{m^2 \varepsilon^2 \sin{\left((t/\varepsilon)\arctan{m \varepsilon}\right)}}{1+m^2 \varepsilon^2}\right] = \\
 &= \sin{\left(\frac{t}{\varepsilon} \arctan{m \varepsilon}\right)}.
\end{align*}

The formula for $\sum\limits_{x\in\z} a_2(x,t,m,\varepsilon)$ is proved analogously.
\end{proof}

\begin{corollary}
For each $t\in\zp$ we have
$$\sum\limits_{x\in\z} a(x,t,m,\varepsilon) = i \, \exp{\left(- i \, \frac{t - \varepsilon}{\varepsilon} \arctan{m \varepsilon}\right)}.$$
\end{corollary}

\begin{proof}
Follows directly from Proposition~\ref{id1} and Euler's formula.
\end{proof}

\ \\

Now, perform a change of the coordinates: $(x, t) \mapsto (\lambda, \mu) = (\frac{t+x}{2\varepsilon}, \frac{t-x}{2\varepsilon})$. We just rotate the coordinate axes through $90\degree$ about zero, and afterwards we scale everything $1/\varepsilon$ times (see Figure~\ref{cov}). 
For positive integers $\lambda, \mu$ denote $b(\lambda,\mu,m,\varepsilon) := a(\varepsilon (\lambda - \mu), \varepsilon (\lambda + \mu), m, \varepsilon)$. We use the notation $b_1(\lambda,\mu,m,\varepsilon)$ and $b_2(\lambda,\mu,m,\varepsilon)$ for the real and the imaginary part of $b(\lambda,\mu,m,\varepsilon)$.

\begin{figure}[H]
\centering
\begin{subfigure}{0.38\textwidth}
\includegraphics[width=1\textwidth]{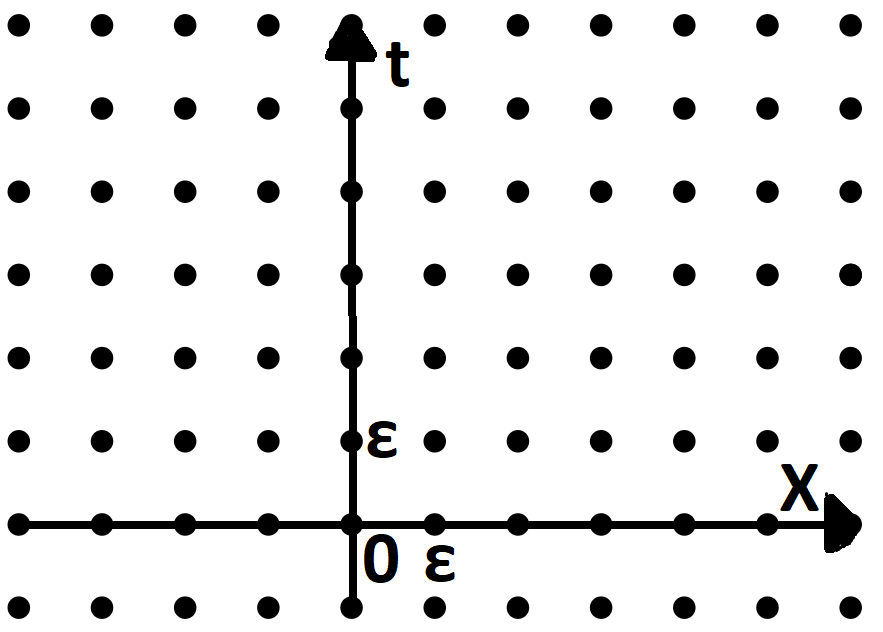}
\end{subfigure}
$\xrightarrow{\text{Change of the coordinates}}$
\begin{subfigure}{0.38\textwidth}
\includegraphics[width=1\textwidth]{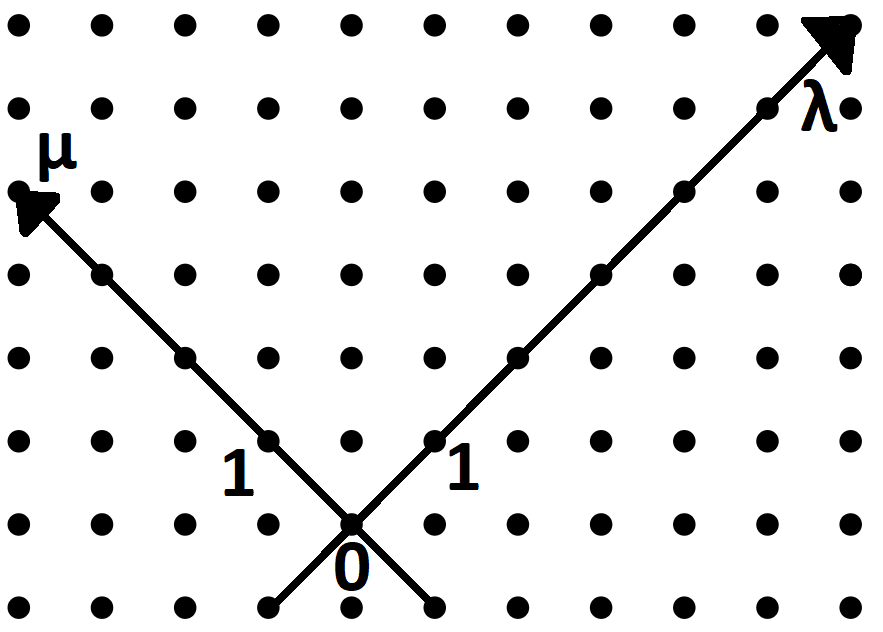}
\end{subfigure}
\caption{}
\label{cov}
\end{figure}

The following table shows $b(\lambda,\mu,m,\varepsilon)$ for small $\lambda$ and $\mu$. The number in a cell $(\lambda,\mu)$ is $b(\lambda,\mu,m,\varepsilon)$. 
(Compare this table with Table~\ref{a-table}.)

\begin{table}[H]
\centering
\begin{tabular}{|c|c|c|c|c|}
\hline
$3$&$\frac{m \varepsilon}{(1+m^2\varepsilon^2)^{3/2}}$&$\frac{(m \varepsilon - 2 m^3 \varepsilon^3) - m^2 \varepsilon^2 i}{(1+m^2 \varepsilon^2)^{2}}$&$\frac{(m \varepsilon - 4m^3 \varepsilon^3 + m^5 \varepsilon^5) + 2 (m^4 \varepsilon^4 - m^2 \varepsilon^2) i}{(1+m^2 \varepsilon^2)^{5/2}}$&$\frac{(m \varepsilon - 6 m^3 \varepsilon^3 + 3m^5 \varepsilon^5) - (3 m^2 \varepsilon^2 - 6 m^4 \varepsilon^4 + m^6 \varepsilon^6) i}{(1+m^2 \varepsilon^2)^{3}}$\\
\hline
$2$&$\frac{m \varepsilon}{1+m^2\varepsilon^2}$&$\frac{(m \varepsilon- m^3 \varepsilon^3) - m^2 \varepsilon^2 i}{(1+m^2\varepsilon^2)^{3/2}}$&$\frac{(m \varepsilon - m^3 \varepsilon^3) + (m^4 \varepsilon^4 - 2 m^2 \varepsilon^2) i}{(1+m^2 \varepsilon^2)^2}$&$\frac{(m \varepsilon - m^3 \varepsilon^3) + 3 (m^4 \varepsilon^4 - m^2 \varepsilon^2) i}{(1+m^2 \varepsilon^2)^{5/2}}$\\
\hline
$1$&$\frac{m\varepsilon}{\sqrt{1+m^2\varepsilon^2}}$&$\frac{m \varepsilon - m^2 \varepsilon^2 i}{1+m^2\varepsilon^2}$&$\frac{m \varepsilon - 2 m^2 \varepsilon^2 i}{(1+m^2\varepsilon^2)^{3/2}}$&$\frac{m \varepsilon - 3 m^2 \varepsilon^2 i}{(1+m^2\varepsilon^2)^{2}}$\\
\hline
$0$&$i$&$\frac{1}{\sqrt{1+m^2 \varepsilon^2}} i$&$\frac{1}{1+m^2 \varepsilon^2} i$&$\frac{1}{(1+m^2 \varepsilon^2)^{3/2}} i$\\
\hline
\diagbox[dir=SW,height=21pt]{$\mu$}{$\lambda$}&$1$&$2$&$3$&$4$ \\
\hline
\end{tabular}
\caption{$b(\lambda,\mu,m,\varepsilon)$ for small $\lambda$ and $\mu$.}
\label{lm-table}
\end{table}

\begin{proposition} \label{sk} \ \\
1) For each fixed $\mu \in \Z_+$ we have
$$ \sum\limits_{\lambda = 1}^{\infty} b_1(\lambda,\mu,m,\varepsilon) = 
(-1)^{\mu+1} \frac{1 + \sqrt{m^2 \varepsilon^2 + 1}}{m \varepsilon} \quad \text{and} \quad
\sum\limits_{\lambda = 1}^{\infty} b_2(\lambda,\mu,m,\varepsilon) =
(-1)^{\mu} \frac{2 + m^2 \varepsilon^2 + 2\sqrt{m^2 \varepsilon^2 + 1}}{m^2 \varepsilon^2}.$$ 
2) For each fixed $\lambda \in \Z_+$ we have
$$ \sum\limits_{\mu = 0}^{\infty} b_1(\lambda,\mu,m,\varepsilon) = 
(-1)^{\lambda+1} \frac{1 + \sqrt{m^2 \varepsilon^2 + 1}}{m \varepsilon} \quad \text{and} \quad
\sum\limits_{\mu = 0}^{\infty} b_2(\lambda,\mu,m,\varepsilon) = 
(-1)^{\lambda+1}.$$
\end{proposition}

Proposition~\ref{sk} easily follows from Proposition~\ref{Dirac}. We omit the proof.

\subsection{Quadratic identities}


\begin{proposition} \ \label{quad1} \\
1) For each fixed $\mu \in \Z_+$ we have:
\begin{flalign*}
\quad \sum\limits_{\lambda = 1}^{\infty} b_1(\lambda,\mu,m,\varepsilon)^2 = 1. &&
\end{flalign*}
2) For each fixed $\lambda \in \Z_+$ we have:
\begin{flalign*}
& \quad a) \sum\limits_{\mu = 0}^{\infty} b_1(\lambda,\mu,m,\varepsilon)^2 = 1; \\
& \quad b) \sum\limits_{\mu = 0}^{\infty} b_2(\lambda,\mu,m,\varepsilon)^2 = 1. &&
\end{flalign*}
\end{proposition}

In the proof of quadratic identities, Proposition~\ref{Dirac} does not help much. To prove Proposition~\ref{quad1}, we need a generalization of Proposition~\ref{prob}. For this purpose, we need an auxiliary definition.

\begin{definition}[{\cite{main}}]
For a set $T\subset\zz$ and a point $(x,t)\in\zz$ with $t>0$, we define $a(x,t\text{ bypass }T;m,\varepsilon)$ analogously to $a(x,t,m,\varepsilon)$, only the sum is over those checker paths  that do not pass through the points of the set $T$. Denote
$P(x,t\text{ bypass }T;m,\varepsilon) = \left|a(x,t\text{ bypass }T;m,\varepsilon)\right|^2$.
We denote by $a_1(x,t\text{ bypass }T;m,\varepsilon)$ and $a_2(x,t\text{ bypass }T;m,\varepsilon)$ the real and the imaginary part of $a(x,t\text{ bypass }T;m,\varepsilon)$ respectively.
\end{definition}


\begin{example}
Let us compute $P(0,4\text{ bypass }(2\varepsilon,2\varepsilon);m,\varepsilon)$. Now, we must not consider the leftmost checker path from Figure~\ref{example1}. Thus 
\begin{align*}
a(0,4\varepsilon \text{ bypass }(2\varepsilon,2\varepsilon);m,\varepsilon) 
&= \frac{ - m^2 \varepsilon^2 - m^3 \varepsilon^3 i}{(1+m^2\varepsilon^2)^{3/2}}; \\
P(0,4\varepsilon \text{ bypass }(2\varepsilon,2\varepsilon);m,\varepsilon) &=  \frac{m^4 \varepsilon^4+ m^6 \varepsilon^6}{(1+m^2\varepsilon^2)^{3}} = \frac{m^4 \varepsilon^4}{(1+m^2\varepsilon^2)^{2}}.
\end{align*}
\end{example}

The following proposition generalizes Proposition~\ref{Dirac}. We do not give the proof because it is essentially the same as the proof of Proposition~\ref{Dirac} (See proof of Proposition~4 from \cite{main}).

\begin{proposition} \label{gen-Dirac}
For each set $T \subset \zz$ and each point $(x,t)\in\zz$ such that $(x,t)\notin T$ and $t\geqslant 2\varepsilon$ we have:
\begin{align*}
1) \: a_1(x,t\text{ bypass }T;m,\varepsilon) &= \frac{1}{\sqrt{1+m^2\varepsilon^2}}(a_1(x+\varepsilon,t-\varepsilon \text{ bypass }T;m,\varepsilon) + m \varepsilon \, a_2(x+\varepsilon,t-\varepsilon \text{ bypass }T;m,\varepsilon);\\
2) \: a_2(x,t\text{ bypass }T;m,\varepsilon) &= \frac{1}{\sqrt{1+m^2\varepsilon^2}}(a_2(x-\varepsilon,t-\varepsilon \text{ bypass }T;m,\varepsilon) - m \varepsilon \, a_1(x-\varepsilon,t-\varepsilon \text{ bypass }T; m,\varepsilon)).
\end{align*}
\end{proposition}

\begin{remark}
If we apply the above proposition for $T = \emptyset$, then we obtain Proposition~\ref{Dirac}.
\end{remark}

The following proposition was first stated and proven by G.~Minaev and I.~Russkikh, but their proof was quite complicated and technical. We present a simple alternative proof.

\begin{proposition}[Generalized probability conservation law, Minaev-Russkikh, private communication] \label{gen-prob}
For each finite set $T\in\zz$ such that there is no infinite checker path from $(0,0)$ bypassing the points of $T$ we have 
$$\sum_{(x,t) \in T} P(x,t\text{ bypass }T\setminus (x,t);m,\varepsilon)=1.$$
\end{proposition}

\begin{remark}
If we apply the above proposition for the set 
$T_t = \{(x, t) \in \zz \,:\, - t \leqslant x \leqslant t)\}$ (where $t$ is fixed), then we obtain Proposition~\ref{prob}.
\end{remark}

\begin{proof}[Proof of Proposition~\ref{gen-prob}] (See Figure~\ref{edges}) Join each point $(x, t) \in \zz$, where $t>0$ and $(t + x) / \varepsilon$ is even, with the points $(x-\varepsilon,t - \varepsilon)$ and $(x+\varepsilon,t - \varepsilon)$. Assign the numbers $a_2(x, t\text{ bypass }T \setminus (x, t); m, \varepsilon)^2$ and $a_1(x, t\text{ bypass }T \setminus (x, t); m, \varepsilon)^2$ respectively to the resulting edges. An edge joining 
$(x \pm \varepsilon, t-\varepsilon)$ and $(x, t)$ is painted red, if $(x,t) \in T$ and $(x \pm \varepsilon, t-\varepsilon) \notin T$ (see Figure~\ref{edg2}). 

It is clear that
$$\sum_{(x,t) \in T} P(x,t\text{ bypass }T\setminus (x,t);m,\varepsilon) =  \sum_{\text{red edges } e} j(e),$$
where $j(e)$ is the number assigned to the edge $e$. Now, the required assertion follows from the following two observations: \\
1) $j((0,0)\text{-}(\varepsilon, \varepsilon)) = 1$ (We denote by $a\text{-}b$ the edge joining points $a$ and $b$);\\
2) (see Figure~\ref{saving}) For each point $(x,t)\in\zz$ (with positive $t$ and even $(x+t)/\varepsilon$) we have
\begin{equation} \label{currents}
j((x,t)\text{-}(x + \varepsilon, t-\varepsilon)) + j((x,t)\text{-}(x - \varepsilon, t-\varepsilon)) = 
j((x,t)\text{-}(x + \varepsilon, t+\varepsilon)) + j((x,t)\text{-}(x - \varepsilon, t+\varepsilon)).
\end{equation}
The latter observation easily follows from Proposition~\ref{gen-Dirac}.
\end{proof}

\begin{remark}
Suppose that in Figure~\ref{saving} there is electrical current, which flows into the point $(x,t)$ from the points $(x \pm \varepsilon,t-\varepsilon)$, and out of the point $(x,t)$ to the points $(x \pm \varepsilon,t + \varepsilon)$. Then equation~\eqref{currents} is just Kirchhoff's current law.
\end{remark}

\begin{figure}[H]
\begin{subfigure}{0.32\textwidth}
\includegraphics[width=1.0\textwidth]{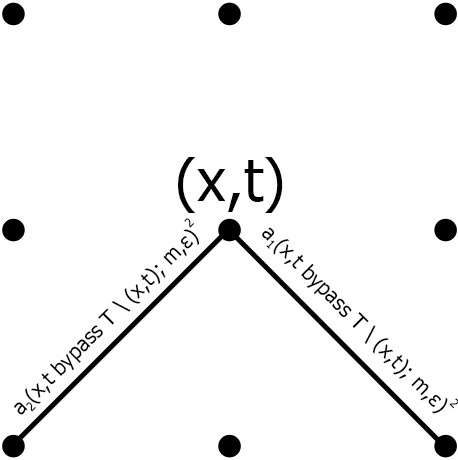}
\caption{}
\label{edges}
\end{subfigure}
\hfill
\begin{subfigure}{0.40\textwidth}
\includegraphics[width=1.0\textwidth]{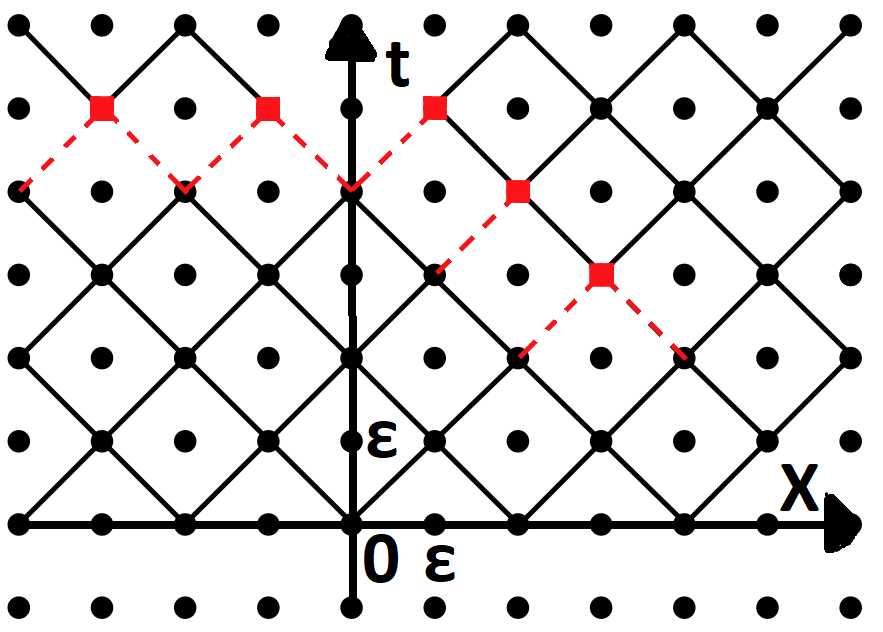}
\caption{}
\label{edg2}
\end{subfigure}
\hfill
\begin{subfigure}{0.24\textwidth}
\includegraphics[width=1.0\textwidth]{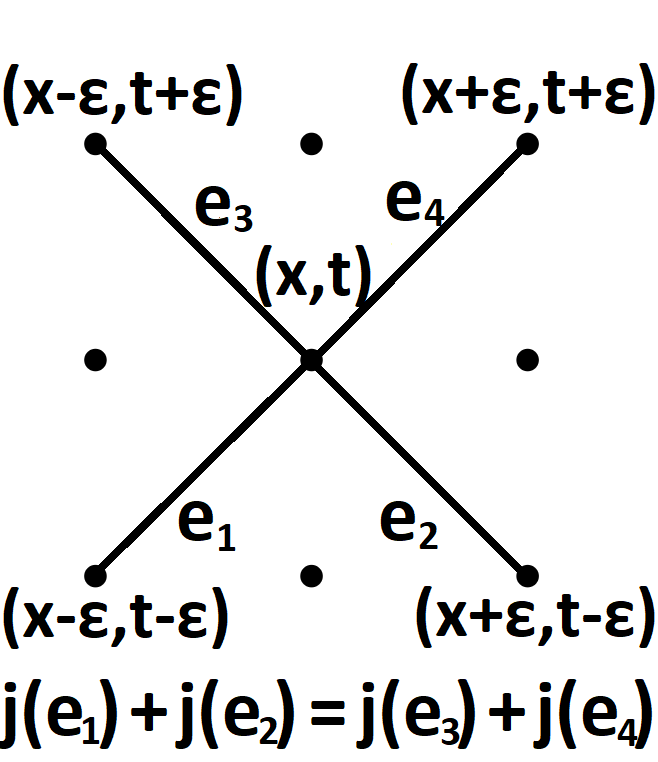}
\caption{}
\label{saving}
\end{subfigure}
\caption{In (b) the red (dash) edges for a particular set $T$ consisting of the red (square) points}
\end{figure}

Using the latter proposition, we finally prove Proposition~\ref{quad1}.

\begin{proof}[Proof of Proposition~\ref{quad1}]
Let us prove assertion 1). Fix $\widehat \mu \in \Z_+$. Consider two sequences of sets of points in $\Z^2$, $S_n$ and $T_n$,
defined as follows:
$S_n = \{ (\lambda, \mu) \in \Z^2  \, | \, 1 \leqslant \lambda \leqslant n, \mu = \widehat \mu \}$, 
$T_n = \{ (\lambda, \mu) \in \Z^2  \, | \, 0 \leqslant \mu < \widehat \mu, \lambda + \mu = n + \widehat \mu \}$. Table~\ref{examples} shows sets $T_n$ and $S_n$ for small $n, \widehat \mu$. By Proposition~\ref{gen-prob} for each $n \in \Z_+$ we have
$$1 = \sum_{(\lambda,\mu) \in T_n \cup S_n} P(\lambda,\mu \text{ bypass }(T_n \cup S_n) \setminus (\lambda,\mu);m,\varepsilon) = 
\sum_{\lambda = 1}^{n} b_1(\lambda, \widehat \mu,m,\varepsilon)^2 + 
\sum_{\mu = 0}^{ \widehat \mu - 1} \left( b_1(n, \mu,m,\varepsilon)^2 + b_2(n, \mu,m,\varepsilon)^2 \right),$$
where the last equality holds because any checker path to a point of $T_n$ cannot pass through other points of $T_n \cup S_n$, and also because a checker path to a point of $S_n$ does not pass through other points of $T_n \cup S_n$ if and only if the checker path finishes with an upwards-left move.

Thus it remains to prove that for each fixed $\mu$ we have $b_1(n, \mu,m,\varepsilon) \to 0$ and $b_2(n, \mu,m,\varepsilon) \to 0$ as $n \to +\infty$. By Proposition~\ref{formula} for $\lambda \geqslant \mu$ we have
$$b_1(n, \mu,m,\varepsilon) = 
\sum\limits_{r=0}^{\mu -1} (1+m^2 \varepsilon^2)^{(1 - n - \mu)/2} (-1)^r (m \varepsilon)^{2r+1} \binom{\mu - 1}{r} \binom{n - 1}{r}.$$

Each summand in the sum tends to $0$ as $n \to +\infty$. Since for each fixed $\mu$ the number of summands is finite, it follows that $b_1(n, \mu,m,\varepsilon) \to 0$ as $n \to +\infty$. Analogously $b_2(n, \mu,m,\varepsilon) \to 0$. Assertion 1) is proved. 

Assertion 2a) follows from 1) by the first identity of Proposition~\ref{symmetry}, and assertion 2b) is proved analogously to 1).
\end{proof}

\begin{table}[H]
\begin{center}
\caption{Sets $T_n$ and $S_n$ for small $n, \widehat \mu$}
\label{examples}
\begin{tabular}{|c| p{4.1cm} | p{4.1cm} | p{4.1cm} | p{4.1cm} |}
\hline
\multicolumn{4}{|c|}{\includegraphics[trim = 0 60 0 -1, width = 0.3\textwidth]{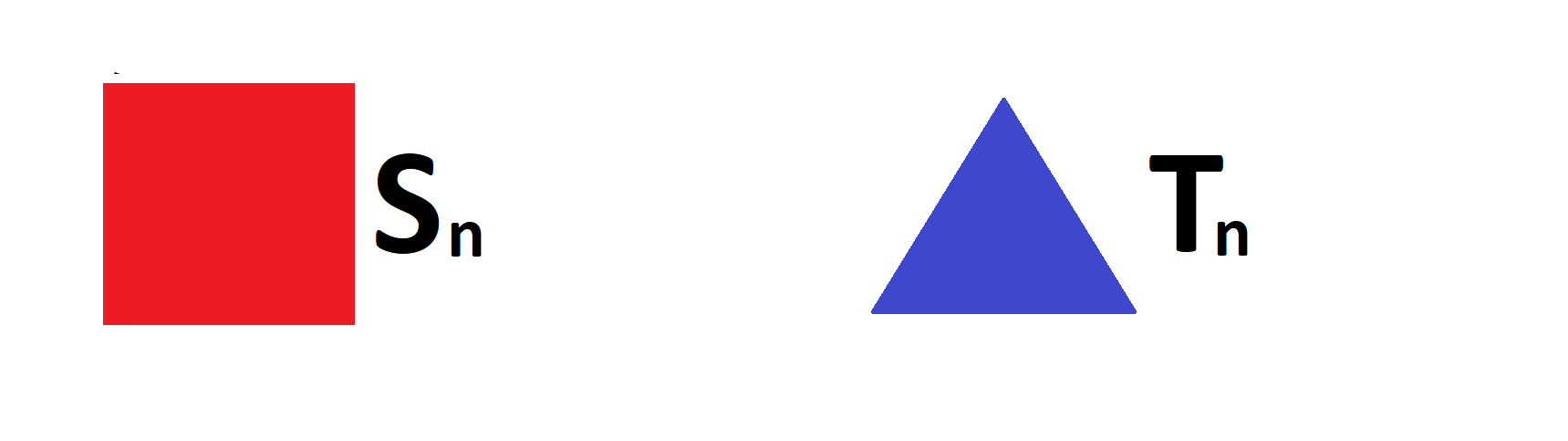}}\\
\hline
$2$
&\includegraphics[trim = 0 0 0 -1, width=0.24\textwidth]{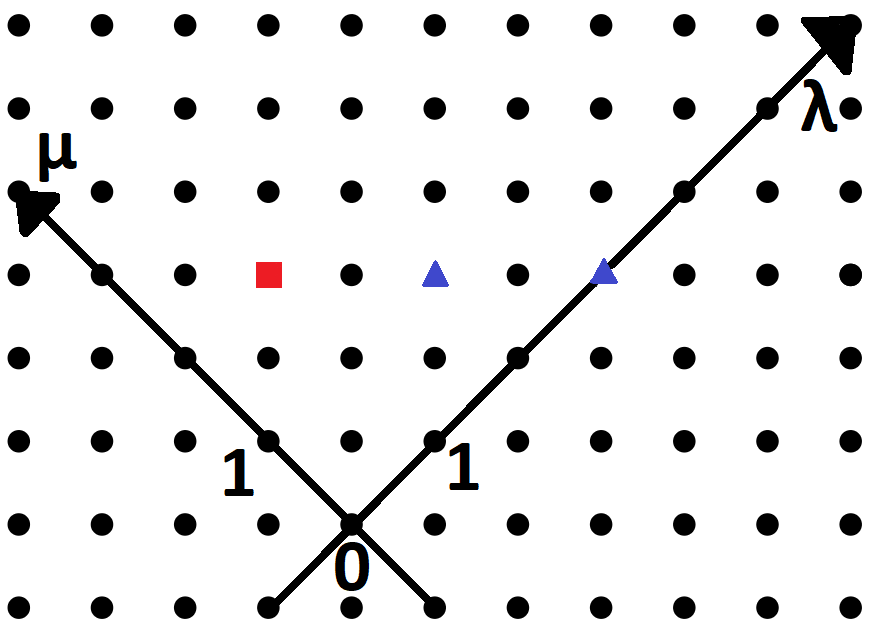}
& \includegraphics[width=0.24\textwidth]{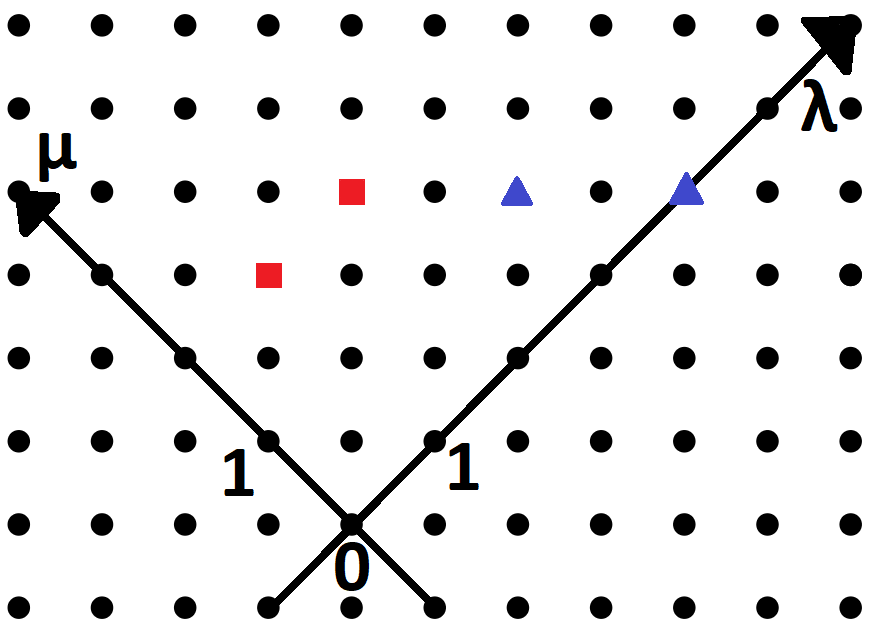}
& \includegraphics[width=0.24\textwidth]{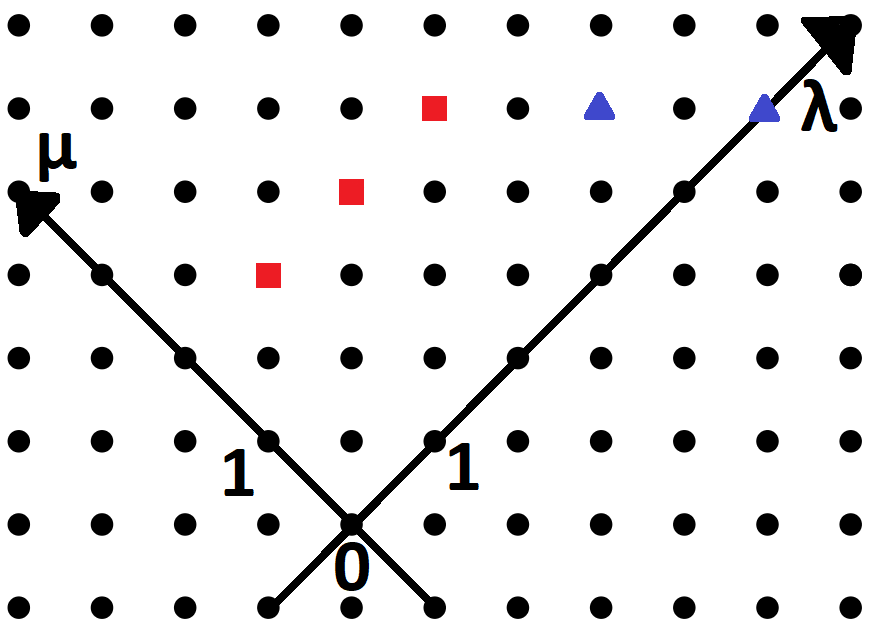} \\
\hline
$1$
&\includegraphics[trim = 0 0 0 -1, width=0.24\textwidth]{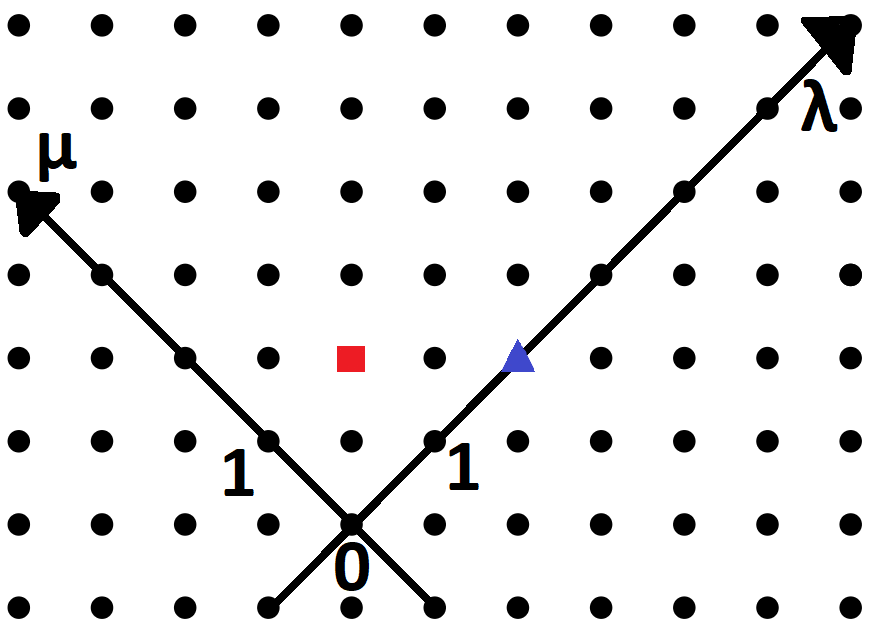}
& \includegraphics[width=0.24\textwidth]{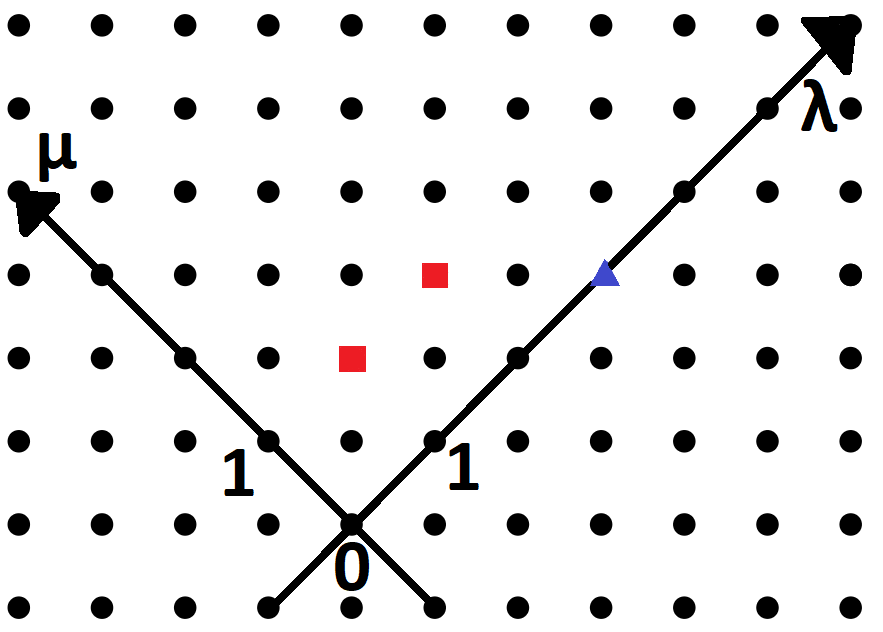}
& \includegraphics[width=0.24\textwidth]{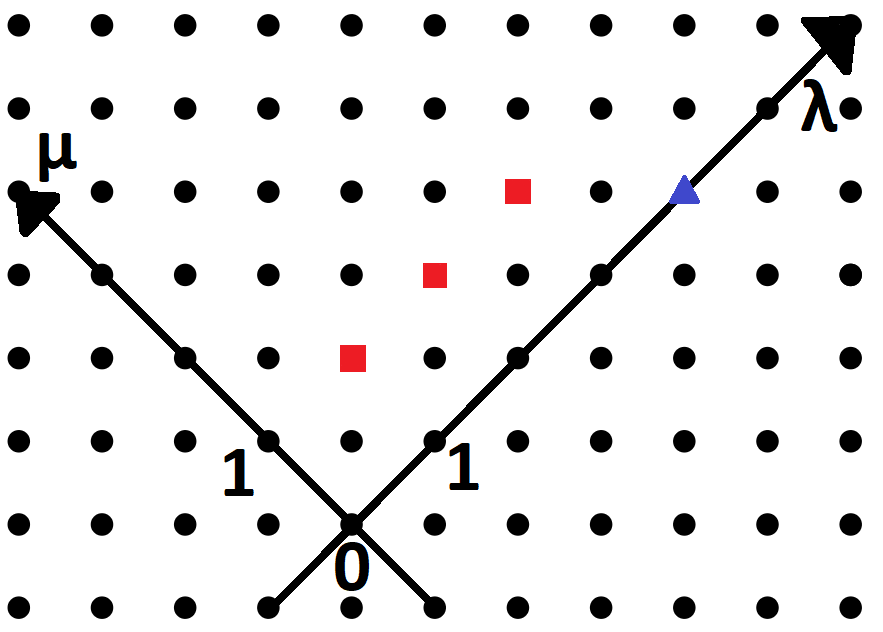} \\
\hline
\diagbox[dir=SW,height=21pt]{$\widehat \mu$}{$n$}&$1$&$2$&$3$ \\
\hline
\end{tabular}
\end{center}
\end{table}

In conclusion, we state a conjecture.

\begin{conjecture}
For each fixed $\mu \in \Z_+$ we have
\begin{flalign*}
&1) \sum\limits_{\lambda = 1}^{\infty} b_2(\lambda,\mu,m,\varepsilon)^2 = \frac{m^2 \varepsilon^2 + 2}{m^2 \varepsilon^2}; \\
&2) \sum\limits_{\lambda = 1}^{\infty} \lambda b_1(\lambda,\mu,1/\varepsilon,\varepsilon)^2 = 3\mu - 1; \\
&3) \sum\limits_{\lambda = 1}^{\infty} \lambda^2 b_1(\lambda,\mu,1/\varepsilon,\varepsilon)^2 = 13\mu^2 - 10 \mu+3; \\
&4) \sum\limits_{\lambda = 1}^{\infty} \frac{1}{\lambda} b_1(\lambda,\mu,1/\varepsilon,\varepsilon)^2 = 2^{\mu-1} \log{2} - \sum\limits_{j=1}^{\mu - 1} \frac{1}{j \cdot 2^j}; \\
&5) \sum\limits_{\lambda = 1}^{\infty} \frac{1}{2^\lambda} b_1(\lambda,\mu,1/\varepsilon,\varepsilon)^2 = \frac{2^{\mu-1} \binom{2\mu-2}{\mu-1}}{3^{2\mu+1}}. &&
\end{flalign*}
\end{conjecture}

\begin{remark}
Despite the fact that the first equality in the conjecture is similar to the equalities in Proposition~\ref{quad1}, it cannot be proved analogously.
\end{remark}

\section{Conclusions}
In this work, we presented a lot of combinatorial identities in Feynman checkers. We used some of these identities to prove the following results: if there is at least one checker path to the point, then the probability to find an electron at this point is non-zero (Theorem~\ref{non-zero}); the expectation of the average electron velocity equals the time-average of the expectation of the instantaneous electron velocity (Proposition~\ref{velocity}). We also found the limit value of the average electron velocity (Theorem~\ref{limit_velocity}). There are several identities that are yet to be proved and possibly there are many interesting identities that are yet to be discovered.

\end{document}